\documentclass[runningheads]{llncs}
\usepackage[lined,ruled,linesnumbered]{algorithm2e}
\usepackage{etex}
\RestyleAlgo{boxruled}
\LinesNumbered
\usepackage[colorlinks=true,breaklinks=true,bookmarks=true,urlcolor=blue,
     citecolor=blue,linkcolor=blue,bookmarksopen=false,draft=false]{hyperref}
\usepackage{framed}
\usepackage{enumitem}
\usepackage{graphicx, tikz}
\usepackage{amsmath,amssymb,amsfonts,amsbsy}
\usepackage{algorithmic}
\usepackage{tabularx,subfigure}
\usepackage{booktabs}
\usepackage{array}
\usepackage{multirow}
\usepackage{float}
\usepackage{empheq}
\usepackage{footnote}
\usepackage{graphicx, subfigure}

\newcommand{\R}{\mathbb{R}}
\newcommand{\N}{\mathbb{N}}

\newcommand{\B}{\mathcal{B}}
\newcommand{\M}{\mathcal{M}}
\newcommand{\F}{\mathcal{F}}

\renewcommand{\mid}{:}

\newcommand{\dmax}{\delta}

\renewcommand{\P}{\mathbb{P}}
\renewcommand{\vec}{\mathbf}

\usepackage[textsize=tiny,textwidth=2.2cm,shadow]{todonotes}

\SetAlFnt{\small}
\SetAlCapFnt{\small}
\SetAlCapNameFnt{\small}
\SetAlCapHSkip{0pt}
\IncMargin{-\parindent}

\spnewtheorem*{assumption}{Assumption}{\itshape}{\rmfamily}
\begin{document}

\title{Resource Competition on Integral Polymatroids}
\titlerunning{Resource Competition on Integral Polymatroids}
\author{Tobias Harks\inst{1} \and Max Klimm\inst{2} \and Britta Peis\inst{3}}

\institute{Department of Quantitative Economics, Maastricht University, the Netherlands
\email{t.harks@maastrichtuniversity.nl}
\and
Department of Mathematics, Technische Universit\"at Berlin, Germany
\email{klimm@math.tu-berlin.de}
\and
School of Business and Economics, RWTH Aachen University, Germany
\email{britta.peis@oms.rwth-aachen.de}
} 
\authorrunning{Harks et al.}
\maketitle

\begin{abstract}
We study competitive resource allocation problems in which players distribute their demands integrally on a set of resources subject to player-specific submodular capacity constraints. Each player has to pay for each unit of demand a cost that is a nondecreasing and convex function of the total allocation of that resource. This general model of resource allocation generalizes both singleton congestion games with integer-splittable demands and matroid congestion games with player-specific costs. As our main result, we show that in such general resource allocation problems a pure Nash equilibrium is guaranteed to exist by giving a pseudo-polynomial algorithm computing a pure Nash equilibrium. 
\end{abstract}




\section{Introduction}
In an influential paper, Rosenthal~\cite{Rosenthal73a} introduced \emph{congestion
games}, a class of strategic games, where
a finite set of players competes over a finite set of
resources.  Each player is associated with a set of allowable subsets of resources
and a pure strategy of a player consists of an allowable subset. 
In the context of \emph{network games}, the resources may correspond to edges of a graph and the
allowable subsets correspond to the paths connecting a source and a sink.
The utility of a resource depends only on the number of
players choosing the same resource and each player wants to maximize (minimize)
the utility (cost) of the sum of the resources contained in the selected subset.
Rosenthal proved the existence of a pure Nash equilibrium. Up to day congestion games
have been used as reference models for describing decentralized systems involving the selfish allocation of congestible resources (e.g., selfish
route choices in traffic networks~\cite{Beckmann56,Roughgarden_Book2005,Wardrop52}
and flow control in telecommunication networks~\cite{Johari06,Kelly98,Srikant03})
and for decades they have been a focal point of research
in (algorithmic) game theory, operations research and theoretical computer science. 
 
In the past, the existence of pure Nash equilibria has been analyzed in many variants
of congestion games such as singleton congestion games with player-specific cost
functions (cf.~\cite{Gairing11,Ieong05,Milchtaich96,Milchtaich06}), congestion games with weighted players 
(cf.~\cite{Ackermann09,Anshelevich08,Chen09,Fotakis05,Harks:existence}), nonatomic and atomic splittable congestion games 
(cf.~\cite{Beckmann56,Haurie85,Kelly98,Wardrop52})
and congestion games with player- and resource-specific and variable demands (cf.~\cite{Harks:TARK}).

Most of these previous works can be classified according to
the following two categories: (i) the demand of each player
is unsplittable and must be completely assigned to exactly one subset
of the allowable subsets; (ii) the demand of a player may be fractionally split
over the set of allowable subsets. While these assumptions and the resulting
models are obviously important (and also realistic for some applications),
they do not allow for the requirement that only \emph{integral fractions} of the demand
may be assigned to allowable subsets of resources. This requirement is 
clearly important in many applications, where the demand represents a collection
of indivisible items or tasks that need to be placed on subsets of resources.
Examples include the scheduling of integer-splittable tasks in the context
of load balancing on server farms (cf.~\cite{KrystaSV03}) or in logistics
where a player controls a fleet of vehicles 
and each must be assigned to a single route. 

Although Rosenthal proposed
congestion games with integer-splittable demands as an important
and meaningful model already back in 1973 -- in his first work on  congestion games \cite{Rosenthal73b}
even published prior to his more famous work~\cite{Rosenthal73a} --  not much is known
regarding existence and computability of pure Nash equilibria.
Rosenthal gave an example showing that in general, pure Nash equilibria
need not exist. Dunkel and Schulz~\cite{Dunkel08} strengthened this result showing that the existence of a pure Nash equilibrium in integer-splittable congestion games is NP-complete to decide. Meyers~\cite{Meyers08} proved that in games with linear cost functions, a pure Nash equilibrium is always guaranteed to exist. For singleton strategy spaces and nonnegative and convex cost functions, Tran-Thanh et al.~\cite{Tran11} showed the existence of pure Nash equilibria. They also showed that
pure Nash equilibria need not exist (even for the restricted strategy spaces)
if cost functions are semi-convex.

\subsubsection{Our Results.}
We introduce congestion games on \emph{integral polymatroids}, where each player may fractionally
assign the demand in integral units among the allowable subsets of resources
subject to player-specific submodular capacity constraints. This way, the resulting strategy space for each player forms an integral polymatroid base polyhedron (truncated at the player-specific demand). 
As our main result, we devise an algorithm that computes a pure Nash equilibrium for 
congestion games on integral polymatroids with player-specific nonnegative, nondecreasing and
\emph{strongly semi-convex} cost functions.  The class of strongly semi-convex functions strictly includes convex functions but is included in the class of semi-convex functions (see Section~\ref{sec:strong_semiconvex} for a formal
definition). The runtime of our algorithm is bounded by $n^{\dmax+1} \cdot m^{\dmax} \dmax^{\dmax+1}$,
where $n$ is the number of players, $m$ the number of resources, and
$\dmax$ is an upper bound on the maximum demand. Thus, for constant $\dmax$, the algorithm
is polynomial.

Our existence result generalizes that of Tran-Thanh et al.~\cite{Tran11} 
for singleton congestion games with integer-splittable demands and convex cost functions and that of
Ackermann et al.~\cite{Ackermann09} for  matroid congestion games with unit demands and player-specific nondecreasing costs. For the important class of network design games, where players need to allocate bandwidth in integral units across multiple spanning trees of a player-specific 
communication graph (cf.~\cite{AntonakopoulosCSZ11,ChenRV10,FalkeHarks13}),
our result shows for the first time the existence of pure Nash equilibria provided
that the cost on each edge is a strongly semi-convex function of total bandwidth allocated.
\subsubsection{Techniques.}
Our algorithm for computing pure Nash equilibria maintains data structures for \emph{preliminary demands, strategy spaces, and strategies}
of the players that all are set to zero initially.
Then, it iteratively increases the demand of a player by one unit
and recomputes a preliminary pure Nash equilibrium (with respect
to the current demands) by following a sequence of best response moves
of players.  The key insight to prove the correctness of the algorithm
is based on two invariants that are fulfilled during the course of the algorithm.
As first invariant we can without loss of generality assume that,
whenever the demand of a player is increased by one unit, there
is a best response that assigns the new unit to some resource \emph{without}
changing the allocation of previously assigned demand units. 
As second invariant, we obtain that, after assigning this new unit 
to some resource, only those players that use
the resource with increased load may have an incentive to deviate. Moreover, there is a best response
that has the property that at most a single unit is shifted to some other resource.
Given the above two invariants, we prove that during the
sequence of best response moves a carefully defined vector of \emph{marginal costs}
lexicographically decreases, thus, ensuring that the sequence is finite.

The first invariant follows by reducing an integral polymatroid
to an ordinary matroid (cf. Helgason~\cite{helgason1974aspects}) and the fact that for a matroid,
a minimum independent set $I_d$ with rank $d$ can be extended to a minimum independent set $I_{d+1}$ with rank $d+1$ by adding a single element to $I_d$.
The second invariant, however, is significantly more complex
since a change of the load of one resource results (when using the matroid construction in the spirit of Helgason) in changed element weights of several elements simultaneously. To prove the second invariant we use several exchange and uncrossing arguments that make use of the submodularity of the rank functions and the fact that a non-optimal basis of a matroid can be improved locally. This is the technically most involved part of our paper.

We note that the above invariants have also been used by Tran-Thanh et al.~\cite{Tran11} 
for showing the existence of pure Nash equilibria in singleton integer-splittable
congestion games. For singleton games, however, these invariants follow almost directly. 
The algorithmic idea to incrementally increase the total demand 
by one unit is similar to the (inductive) existence proof of Milchtaich~\cite{Milchtaich96}
for singleton congestion games with
player-specific cost functions (see also Ackerman et al.~\cite{Ackermann09}
for a similar proof for matroid congestion games).
The convergence proof for our algorithm and the above mentioned invariants, however, are
considerably more involved for general integral polymatroids.

Besides providing new existence results for an important and large class of games, the main contribution of this paper is to propose a unified approach to prove the existence of pure Nash equilibria that connects the seemingly unrelated existence results of Milchtaich~\cite{Milchtaich96} and Ackerman et al.~\cite{Ackermann09} on the one hand, and Tran-Thanh et al.~\cite{Tran11} on the other hand.

\section{Preliminaries}
In this section, we introduce polymatroids, strong semi-convexity, and congestion games on integral polymatroids.

\subsubsection{Polymatroids.}
\label{sec:polymatroids}
Let $\N$ denote the set of nonnegative integers and let $R$ be a finite and non-empty set of resources. We write $\N^R$ shorthand for $\N^{|R|}$. Throughout this paper, vectors $\vec x = (x_r)_{r \in R}$ will be denoted with bold face. An integral (set) function $f : 2^R \rightarrow \N$ is \emph{submodular} if $f(U) + f(V) \geq f(U \cup V) + f(U \cap V)$ for all $U,V \in 2^R$. Function $f$ is \emph{monotone} if
$U\subseteq V$ implies $f(U)\le f(V)$, and \emph{normalized} if $f(\emptyset)=0$.
An integral submodular, monotone and normalized function $f : 2^R \rightarrow \N$ is called an \emph{integral polymatroid rank function}. The associated \emph{integral polyhedron} is defined as
\begin{align*}
\P_f &:= \Bigl\{\vec x \in \N^{R} \mid \sum_{r \in U}x_r\leq f(U)\text{ for each }U\subseteq R \Bigr\}. 
\intertext{Given the integral polyhedron $\P_f$ and some integer $d \in \N$ with $d \leq f(R)$, the \emph{$d$-truncated integral polymatroid} $\P_f(d)$ is  defined as}
\P_f(d) &:= \Bigl\{\vec x\in \N^{R} \mid \sum_{r \in U}x_r\leq f(U)\text{ for each }U\subseteq R, \sum_{r\in R} x_r\le d \Bigr\}.
\intertext{The corresponding \emph{integral polymatroid base polyhedron} is}
\B_f(d) &:= \Bigl\{\vec x\in \N^{R} \mid \sum_{r \in U}x_r\leq f(U)\text{ for each }U\subseteq R, \sum_{r\in R} x_r=d \Bigr\}.
\end{align*}

\subsubsection{Strongly Semi-Convex Functions.}
\label{sec:strong_semiconvex}

Recall that a function $c : \N \to \N$ is \emph{convex} if $c(x+1) - c(x) \leq c(x+2) - c(x+1)$ for all $x \in \N$. 
A function $c$ is called \emph{semi-convex} if the function $x \cdot c(x)$ is convex, i.e., $$(x+1)c(x+1) - xc(x) \leq 
(x+2)c(x+2) - (x+1)c(x+1)$$ for all $x \in \N$.

For the main existence result of this paper, we require a property
 of each cost function that we call \emph{strong semi-convexity}, which is weaker than convexity 
but stronger than semi-convexity. Roughly speaking, it states that the 
marginal difference $c(a+x)x - c(a+x-1)(x-1)$ does not decrease as $a$ or $x$ increase. We also introduce a slightly weaker 
notion, termed $u$-truncated strong semi-convexity, where strong semi-convexity is only required for values of $x$ not 
larger than $u$. 
\begin{definition}[Strong Semi-Convexity]
\label{sem-conv}
A function $c: \N \to \N$ is \emph{strongly semi-convex} if
\begin{equation}\label{semi-conv}
c(a+x)x-c(a+x-1)(x-1) \leq c(b+y)y - c(b+y-1)(y-1)
\end{equation}
for all $x,y \in \N$ with $1 \leq x \leq y$ and all $a,b \in \N$ with $a \leq b$.
For an integer $u \geq 1$, $c$ is \emph{$u$-truncated strongly semi-convex}, if the above inequality is only required to be satisfied for all $x,y \in \N$ with $1 \leq x \leq y \leq u$ and all $a,b \in \N$ with $a \leq b$.
\end{definition}
We note that a similar definition is also given in Tran-Thanh et al.~\cite{Tran11}.
We proceed to prove that strong semi-convexity is indeed strictly weaker than convexity. The proof is moved to the appendix.
\begin{proposition}
\label{pro:semi_convex}
Every convex and nondecreasing function $c : \N \to \N$ is also strongly semi-convex, but not vice versa.
\end{proposition}
\begin{remark}
A function is $1$-truncated strongly semi-convex if and only if it is nondecreasing.
\end{remark}


\subsubsection{Congestion Games on Integral Polymatroids.}
\label{sec:game}
In a congestion game on integral polymatroids, 
there is a non-empty and finite set $N$ of players and a non-empty and finite set
$R$ of resources
each endowed with a
 player-specific cost function $c_{i,r} : \N \to \N$, $r \in R, i \in N$. Each player~$i$ is associated with a demand $d_i \in \N$, $d_i \geq 1$ and an integral polymatroid rank  function $f^{(i)} : 2^R \to \N$ that together define a $d_i$-truncated integral polymatroid $\P_{f^{(i)}}(d_i)$ with base polyhedron $\B_{f^{(i)}}(d_i)$ on the set of resources. A strategy of
player~$i \in N$ is to choose a vector $\vec x_i  = (x_{i,r})_{r \in R} \in \B_{f^{(i)}}(d_i)$, i.e., player $i$ chooses an integral resource consumption $x_{i,r} \in \N$ for each resource $r$
such that the demand $d_i$ is exactly distributed among the resources and for each $U \subseteq R$ not more than $f^{(i)}(U)$ units of demand are distributed to the resources contained in $U$.
Using the notation $\vec x_i=(x_{i,r})_{r \in R}$, the set $X_i$ of feasible strategies of player~$i$ is
defined as
\begin{align*}
X_i = \B_{f^{(i)}}(d_i)
= \Bigl\{\vec x_i \in \N^{R} \mid \sum_{r\in U} x_{i,r} \le f^{(i)}(U) \text{ for each } U \subseteq 
R, \sum_{r\in R} x_{i,r}=d_i \Bigr\}.
\end{align*}
The Cartesian product $X = \bigtimes_{i \in N} X_i$ of the players' sets of feasible
strategies is the joint strategy space. An
element $\vec x = (\vec x_i)_{i \in N} \in X$ is a strategy profile. For a resource $r$, and a strategy profile $\vec x 
\in X$, we write $x_r = \sum_{i \in N}
x_{i,r}$. The private cost of
player~$i$ under strategy profile $\vec x \in X$ is defined as 
$ \pi_i(\vec x) = \sum_{r \in R}
c_{i,r}(x_r) \, x_{i,r}.$
In the remainder of the paper, we will compactly
represent the strategic game by the tuple 
$ G=(N,X,(d_i)_{i\in N},(c_{i,r})_{i\in N,r\in R}).$

We use standard game theory notation. For a player $i \in N$ and a strategy profile $\vec x \in X$,
we write $\vec x$ as $(\vec x_i,\vec x_{-i})$. A \emph{best response} of player~$i$ to $\vec x_{-i}$ is a strategy $\vec 
x_i \in X_i$ with $\pi_i(\vec x_i, \vec{x}_{-i}) \leq \pi_i(\vec y_i, \vec{x}_{-i})$ for all $\vec y_i \in X_i$.
A pure Nash equilibrium is a strategy profile $\vec x \in X$ such that for each player~$i$ the strategy $\vec x_i$ is a 
best response to $\vec x_{-i}$.

Throughout this paper, we assume that the player-specific cost functions of player~$i$ on resource $r$ are $u_{i,r}$-truncated 
strongly semi-convex, where $u_{i,r} = f^{(i)}(\{r\})$. Note that $u_{i,r}$ is a natural upper bound on the units of demand 
player~$i$ can allocate to resource $r$ in any strategy $\vec x_i \in X_i$.

\begin{assumption}
\label{ass:cost}
For all $i\in N, r\in R$, the cost function $c_{i,r} : \N \to \N$ is nonnegative, nondecreasing and $u_{i,r}$-truncated strongly semi-convex, where $u_{i,r} = f^{(i)}(\{r\})$.
\end{assumption}

\subsection{Examples}\label{apps:examples}

We proceed to illustrate that we obtain the well known classes of integer-splittable singleton congestion games and matroid congestion games as special cases of congestion games on integer polymatroids.

\begin{example}[Singleton integer-splittable congestion games]
For the special case that for each player~$i$, there is a player-specific subset $R_i \subseteq R$ of resources such that $f^{(i)}(\{r\}) = d_i$, if $r \in R_i$, and $f^{(i)}(\{r\}) = 0$, otherwise, we obtain integer-splittable singleton congestion games previously studied by Tran-Thanh et al.~\cite{Tran11}. While they consider the special case of convex and \emph{player-independent} cost functions, our general existence result implies existence of a pure Nash equilibrium even for \emph{player-specific} and strongly semi-convex cost functions. 
\end{example}

\begin{example}[Matroid congestion games with player-specific costs]
For the special case, that for each player~$i$, $f^{(i)}$ is the rank function of a player-specific matroid defined on $R$, and $d_i=f^{(i)}(R)$, we obtain ordinary matroid congestion games with player-specific costs and unit demands studied by Ackermann et al.~\cite{Ackermann09} as a special case.

Note that the rank function~$\text{rk}$ of a matroid is always subcardinal, i.e., $\text{rk}(U) \leq |U|$ for all $U \subseteq R$. Thus, we obtain in particular that $\text{rk}(\{r\}) \leq 1$ for all $r \in R$. This implies that our existence result continues to hold if we only require that the player-specific cost functions are $1$-truncated strongly semi-convex, which is equivalent to requiring that cost functions are nondecreasing as in \cite{Ackermann09}. Like this, we obtain the existence result of \cite{Ackermann09} as a special case of our existence result for congestion games on integer polymatroids.
As a strict generalization, our model includes the case in which players
have a demand $d_i\in \N$ that can be distributed in integer units over
bases (or even arbitrary independent sets) of a given player-specific matroid. 
A prominent application arises
in network design (cf.~\cite{AntonakopoulosCSZ11,ChenRV10,FalkeHarks13}), where a player needs to allocate bandwidth along several spanning trees and the cost function for installing enough capacity on an edge is a convex function of the total  bandwidth allocated.
\end{example}

\section{Equilibrium Existence}
In this section, we give an algorithm that computes
a pure Nash equilibrium for congestion games on integral polymatroids.
Our algorithm relies on two key sensitivity properties of optimal solutions minimizing
a linear function over an integral polymatroid base polyhedron (see Lemma~\ref{lem:dem}
and Lemma~\ref{lem:load} below).
After reducing the usual reduction of integer polymatroids to ordinary matroids (cf.\ \cite{helgason1974aspects}),
Lemma~\ref{lem:dem} follows more or less directly from the respective property for matroids.
The proof of Lemma~\ref{lem:load} is considerably more involved and relies heavily on uncrossing arguments.
The two lemmata will be proven formally in Section~\ref{sec:supermatroids}.

\subsubsection{Key Sensitivity Results.}
For two vectors $\vec x_i, \vec y_i\in\N^R$, we denote
their Hamming distance by $H(\vec x_i, \vec y_i):=\sum_{r \in R}|x_{i,r}-y_{i,r}|$.
Lemma~\ref{lem:dem}
 shows that whenever $\vec x_i$ 
minimizes the cost of player $i$ over the base polyhedron $\B_{f^{(i)}}(d_i)$, then 
we only need to increase $x_{i,r}$ for some $r\in R$ by one unit
to obtain a $\vec y_i$ minimizing the player's cost over the base polyhedron $B_{f^{(i)}}(d_i+1)$.

\begin{lemma}[Demand Increase]\label{lem:dem}
Let $\vec x_i \in \B_{f^{(i)}}(d_i)$ be a best response of player~$i$ to $\vec x_{-i} \in X_{-i}$. Then there exists a best 
response $\vec y_i \in \B_{f^{(i)}}(d_i+1)$ to $\vec x_{-i}$ such that
$H(\vec x_i, \vec y_i)=1$.
\end{lemma}
The second result shows that if some other player $j\neq i$
increases the $j$-load on some resource $r$ that is also used by player $i$
with at least one unit, then
player $i$ can simply shift one unit from $r$ to some  $s\in R$
in order to retain minimal costs. 
\begin{lemma}[Load Increase]\label{lem:load}
Let $\vec x_i \in \B_{f^{(i)}}(d_i)$ be a best response of player~$i$ to $\vec x_{-i}\in X_{-i}$ and for each resource $r$ let $a_r = \sum_{j \neq i} x_{j,r}$ be the induced allocation.
If for a resource $r$, the value $a_{r}$ is increase by $1$, then
there exists a best response $\vec y_i\in \B_{f^{(i)}}(d_i)$ towards the new profile
with $H(\vec x_i, \vec y_i) \in \{0,2\}$.
\end{lemma}

\subsubsection{The Algorithm.}
Both sensitivity results are used as the main
building blocks for Algorithm~\ref{alg:greedy}
that computes a pure Nash equilibrium
for congestion games on integral polymatroids.
Algorithm~\ref{alg:greedy}
maintains \emph{preliminary demands, strategy spaces, and strategies} of the players
denoted by $\bar d_i\leq d_i$, $\bar X_i=X_i(\bar d_i)$, and $\vec x_i\in \bar X_i$, respectively. Initially,  $\bar{d}_i$ is set to zero for all $i\in N$ and the strategy profile, where the strategy of each player equals the zero vector is a pure Nash equilibrium for this game in which the demand of each player is zero.

Then, in each round, for some player~$i$ the demand is increased from 
$\bar d_i$ to $\bar d_i+1$, and a best response $\vec y_i\in X(\bar d_i+1)$ with $H(\vec x_i, \vec y_i)=1$ is computed, see Line~\ref{it:comp_demand} in Algorithm~\ref{alg:greedy}. By Lemma~\ref{lem:dem},
such a best response always exists. In effect, the load on exactly one resource $r$
increases and only those players~$j$ with $x_{j,r} > 0$ on this
resource can potentially decrease their private cost by a deviation.
By Lemma~\ref{lem:load}, a best response of such players  consists w.l.o.g.
of moving a single unit from this resource to another resource, see Line~\ref{it:choose_yi} of Algorithm~\ref{alg:greedy}.
As a consequence, during the while-loop (Lines~\ref{it:if}-\ref{it:endif}),
only one additional unit (compared to the previous iteration)
is moved preserving the invariant that only players using a resource
to which this additional unit is assigned may have an incentive to profitably deviate.
Thus, if the while-loop is left, the current strategy profile~$\vec x$
is a pure Nash equilibrium for the reduced game $\bar{G} = (N, \bar X, \bar d,(c_{i,r})_{i \in N, r \in R})$.
\IncMargin{2.5em}
\begin{algorithm}[0.8\textwidth]
 \caption{Compute PNE}
 \label{alg:greedy}
 \Indm\Indmm
    \KwIn{$G=(N,X,(d_i)_{i\in N},(c_{i,r})_{i\in N,r\in R})$}
  \KwOut{pure Nash equilibrium $\vec x$}
  \Indp\Indpp
$\bar d_i \leftarrow 0, \bar X_i\leftarrow X_i(0)$ and $\vec x_i\leftarrow \vec 0$ for all $i\in N$\;
		\For{$k=1,\dots, \sum_{i\in N}d_i$}{	
			Choose $i\in N$ with $\bar d_i<d_i$\;\label{it:choose_player}
			$\bar d_i\leftarrow \bar d_i+1$; $\bar X_i\leftarrow X_i(\bar d_i)$\;\label{it:set_increase_demand}
	Choose a best response $y_i \in \bar{X}_i$ with $H(y_i,x_i) = 1$\;\label{it:comp_demand}
$\vec x_i \leftarrow  \vec y_i$\; \label{it:increase_demand}
\While{$\exists i\in N$ who can improve in $\bar G = (N, \bar X, \bar d,(c_{i,r})_{i \in N, r \in R})$
\label{it:if}}
 {Compute a best response $\vec y_i \in \bar X_i$ with $H(\vec y_i, \vec x_i)=2$\;  \label{it:choose_yi}
$\vec x_i \leftarrow  \vec y_i$\;}
\label{it:endif}}
Return $\vec x$\;
\end{algorithm}

Now we are ready to prove the main existence result.
\begin{theorem}
\label{thm:main}
Congestion games on integral polymatroids
with player-specific nonnegative, nondecreasing, and strongly semi-convex cost functions possess a pure Nash equilibrium.
\end{theorem}
\begin{proof}
We prove by induction on the total demand $d = \sum_{i\in N}d_i$ of the input game $G = (N, X, (d_i)_{i \in N}, 
(c_{i,r})_{i \in N, r \in R})$ that Algorithm~\ref{alg:greedy}
computes a pure Nash equilibrium of $G$.

For $d = 0$, this is trivial.
Suppose that the algorithm works correctly for games with total demand $d-1$ for some $d \geq 1$ and consider a game $G$ with total demand $d$.
Let us assume that in Line~\ref{it:choose_player}, the algorithm always chooses a player with minimum index. Consider the game $G' = (N,X,(d'_i)_{i\in N}, (c_{i,r})_{i \in N, r \in R})$ that differs from $G$ only in the fact that the demand of the last player~$n$ is reduced by one, i.e. $d'_i = d_i$ for all $i<n$ and $d_n' = d_n - 1$. Then, when running the algorithm with $G'$ as input, the $d-1$ iterations (of the for-loop) are equal to the first $d-1$ iterations when running the algorithm with $G$ as input. Thus, with $G$ as input, we may assume that after the first $d-1$ iterations, the preliminary strategy profile that we denote by $\vec x'$ is a pure Nash equilibrium of $G'$.

We analyze the final iteration $k=d$ of the algorithm in which the demand of player~$n$ is increased by $1$ (see Line~\ref{it:set_increase_demand}).
In Line~\ref{it:comp_demand}, a best reply $\vec y_n$ with $H(\vec x_n,\vec y_n) = 1$
is computed which exists by Lemma~\ref{lem:dem}. Then, as long as there is a player~$i$ that can improve unilaterally, in Line~\ref{it:choose_yi}, a best response $\vec y_i$ with $H(\vec y_i, \vec x_i) = 2$ is computed which exists by Lemma~\ref{lem:load}.

It remains to show that the while-loop in Lines~\ref{it:if}--\ref{it:endif} terminates.
To prove this, we give each unit of demand of each player~$i \in N$ an identity
denoted by $i_j, j=1,\dots,d_i$. For a strategy profile $\vec x$, we define
$r(i_j,\vec x)\in R$ to be the resource to which  unit $i_j$ is assigned in strategy profile $\vec x$.
Let $\vec x^l$ be the strategy profile after Line~\ref{it:choose_yi} of the algorithm has been executed the $l$-th time, where we use the convention that $\vec x^0$ denotes the preliminary strategy profile when entering the while-loop. As we chose in Line~\ref{it:comp_demand} a strategy of player~$n$ with Hamming distance one, there is a unique resource $r_0$ such that $x_{r_0}^0 = x'_{r_0} + 1$ and $x_{r}^0 = x_{r}'$ for all $r \in R \setminus \{r_0\}$. Furthermore, because we choose in Line~\ref{it:choose_yi}
 a best response with Hamming distance two, a simple inductive claim shows that after each iteration~$l$ of the while-loop, there is a unique resource $r_l \in R$ such that $x_{r_l}^l = x_{r_l}' + 1$ and $x_{r}^l = x_{r}'$ for all $r \in R \setminus \{r_l\}$.

For any $\vec x^{l}$ during the course of the algorithm, we define the
\emph{marginal cost} of unit $i_j$
under strategy profile $\vec x^{l}$ as
\begin{align}\label{def:delta}
\Delta_{i_j}(\vec x^l)=
\begin{cases} 
c_{i,r}(x_r^l)\,x_{i,r}^l-c_{i,r}(x_r^l-1)\,(x_{i,r}^l-1), & \text{if }r=r(i_j,\vec x)=r_{l}\\
c_{i,r}(x_r^l+1)\,x_{i,r}^l-c_{i,r}(x_r^l)\,(x_{i,r}^l-1), & \text{if }r=r(i_j,\vec x)\neq r_{l}.\end{cases}
\end{align}
Intuitively, if $r(i_j,\vec x)=r_{l}$, the value $\Delta_{i_j}(\vec x)$ measures the \emph{cost saving} on resource 
$r(i_j,\vec x)$ if $i_j$ (or any other unit of player $i$ on resource $r(i_j,\vec x)$)
is removed from $r(i_j,\vec x)$.
If $r(i_j,\vec x)\neq r_{l}$,
the value $\Delta_{i_j}(\vec x)$ measures
the cost saving if  $i_j$ is removed from $r(i_j,\vec x)$ after the total allocation 
has been increased by one unit by some other player.
For a strategy profile $\vec x$
we define $\Delta(\vec x)=(\Delta_{i_j}(\vec x))_{i=1,\dots,n, j=1,\dots, d_i}$ to be the vector
of marginal costs and let
$\bar \Delta(\vec x)$ be the vector of marginal costs sorted in non-increasing order.
We claim that $\bar \Delta(\vec x)$ decreases lexicographically
during the while-loop.
To see this, consider an iteration $l$ in which some unit $i_j$
of player $i$
is moved from resource $r_{l-1}$ to resource $r_{l}$.

For proving $\bar \Delta(\vec x^{l})<_{\text{lex}} \bar \Delta(\vec x^{l-1})$,
we first observe that 
we only have to care for $\Delta$-values that correspond to units $i_j$ of the deviating player~$i$, because for all players $h\neq i$ we obtain
$\Delta_{h_j}(\vec x^{l-1})=\Delta_{h_j}(\vec x^{l})$ 
for all $j=1,\dots, d_h$.
This follows immediately if 
$h_j$ is neither assigned to $r_{l-1}$ nor to $r_{l}$.
If $h_j$ is assigned to $r_{l-1}$ or $r_{l}$,
then we switch the case in \eqref{def:delta},
and the claimed equality still holds.
It remains to consider the $\Delta$-values corresponding to the units of the deviating player~$i$. Recall that the deviation of player~$i$ consists of moving unit $i_j$ from resource $r_{l-1}$ to resource $r_{l}$.
We obtain
\begin{align*}
\Delta_{i_j}(\vec x^{l-1}) &=c_{i,r_{l-1}}(x_{r_{l-1}}^l)\,x_{i,r_{l-1}}^l-c_{i,r_{l-1}}(x_{r_{l-1}}^l-1)\,(x_{i,r_{l-1}}^l-1)
\\&>c_{i,r_{l}}(x_{r_{l}}^l+1)\,(x_{i,r_{l}}^l+1)-c_{i,r_{l}}(x_{r_{l}}^l)\,x_{i,r_{l}}^l =  \Delta_{i_j}(\vec x^{l}),
\end{align*}
where the inequality follows since player $i$ strictly improves.
For every unit $i_m$ of player~$i$ that is assigned to resource $r_l$ as well, i.e, $ r(i_m,\vec x^{l})= r(i_j,\vec x^{l})=r_{l}$, we have
 $\Delta_{i_j}(\vec x^{l})=\Delta_{i_m}(\vec x^{l})$
 since the $\Delta$-value is the same for all units of a single player
assigned to the same resource. The $\Delta$-values
of such units $i_m$ might have increased, but only to the $\Delta$-value of unit $i_j$.

Next, consider the $\Delta$-values of a unit $i_m$ assigned to resource $r_{l-1}$, i.e., $ r(i_m,\vec x^{l})= r(i_j,\vec x^{l})=r_{l-1}$.
We obtain
\begin{align*}
\Delta_{i_m}(\vec x^{l})&=c_{i,r_l}(x_{r_{l-1}}^l)\,(x_{i,r_{l-1}}^l-1)-c_{i,r_{l-1}}(x_{r_{l-1}}^l-1)\,(x_{i,r_{l-1}}^l-2)\\
&\leq c_{i,r_{l-1}}(x_{r_{l-1}}^l)\,x_{i,r_{l-1}}^l-c_{i,r_{l-1}}(x_{r_{l-1}}^l-1)\,(x_{i,r_{l-1}}^l-1)
 = \Delta_{i_m}(\vec x^{l-1}),
\end{align*}
where for the inequality we used that $c_{i,r}(x_{r_{l-1}})\geq c_{i,r}(x_{r_{l-1}}-1)$ as $c_{i,r}$ is nondecreasing.

Altogether, the $\Delta$-values of all units of all players $h\neq i$
have not changed, for player $i$, the  $\Delta$-values of remaining units assigned to resource $r_{l-1}$
 decreased, and the $\Delta$-values assigned to resource $r_{l}$ increased exactly to $\Delta_{i_j}(\vec x^{l})$
which is strictly smaller than $\Delta_{i_j}(\vec x^{l-1})$.
Thus, $\bar \Delta(\vec x^{l})<_{\text{lex}} \bar \Delta(\vec x^{l-1})$ follows.
\qed\end{proof}
The following corollary states an upper bound on the number of iterations of the algorithm in terms of $\delta = \max_{i \in N} d_i$.
The proof can be found in Appendix~\ref{apx:cor}.
\begin{corollary}\label{cor:runtime}
The number of iterations is at most $n^{\dmax+1}m^{\dmax} \dmax^{\dmax+1}$, which yields a polynomial algorithm computing a pure Nash equilibrium for constant $\dmax$. 
\end{corollary}

\section{Sensitivity Analysis for Integral Polymatroids}\label{sec:supermatroids}

It remains to show the key sensitivity results of Lemma~\ref{lem:dem} and Lemma~\ref{lem:load}. For ease of notation, let us drop the index~$i$ form the statements of the lemmata and let us consider a fixed integral polymatroid base polyhedron
$$\mathcal{B}_{f}(d)
= \Bigl\{\vec x\in \N^{R} \mid \sum_{r\in U} x_{r} \le f(U) \text{ for each } U \subseteq 
R, \sum_{r\in R} x_{r}=d \Bigr\}.$$
w.r.t.\ some submodular, monotone, and normalized function $f:2^R\to \N$, and some demand value $d\in N$.

We identify the points in $\mathcal{B}_{f}(d)$ with a set family $\mathcal{F}(d)$ on a largely extended ground set $E$ as follows:
For each resource $r\in R$, let $u_{r}=f(\{r\})$ and let $K_r=\{r_1\prec \ldots \prec r_{u_r}\}$ be a totally ordered set (chain) with $|u_r|$ distinct elements $r_1,\dots,r_{u_r}$.
Let further $E=\bigcup_{r\in R} K_r$ be the disjoint union of these chains. Then, $P = (E,\preceq)$ is a partially ordered set (poset) where two elements $e,e'$ are comparable if and only if they are contained in the same chain $K_r$ for some $r \in R$.
 Furthermore, let $\mathcal{D}(P)$ denote the set of \emph{ideals} of $P$, i.e.,
 $\mathcal{D}(P)$ consists of all subsets $I\subseteq E$ such that for each $e\in I$ all elements $g\prec e$ also belong to $I$.
 Note that there is a one-to-one correspondence between the sets in $\mathcal{D}(P)$ and the integral points in
 $\{\vec x\in \N^R\mid x_r\le f(\{r\})\ \forall r\in R\}$.
 As a consequence, the feasible points in the integral polymatroid $\mathbb{P}_f$ can be identified with the set family

\begin{align}
\F &:=\Bigl\{F\in \mathcal{D}(P)\mid \Bigl|\bigcup_{r\in U} K_r \cap F\Bigr| \le f(U) \text{ for each } U\subseteq R\Bigr\}.
\label{F}
\intertext{Accordingly, the vectors contained in the polymatroid base polyhedron  $\mathcal{B}_f(d)$  for $d\in \N$ can be identified with the set family}
\mathcal{F}(d) &:=\Bigl\{F\in \F\mid |F|=d \Bigr\}.\label{B}
\end{align}

In fact, it is known (see, e.g., \cite{schrijver2003combinatorial} and \cite{helgason1974aspects}) that any integral polymatroid $\mathbb{P}_f$ can be reduced to an ordinary matroid
$\mathcal{M}=(E,r)$ on ground set $E$ with rank function $r:2^E\to \N$ defined via
$$r(U):=\min_{T\subseteq R} \Bigl(\Bigl|U\setminus{\bigcup_{r\in T} K_r}\Bigr| + f(T) \Bigr)$$
for all $U\subseteq E$.
It turns out that the independent sets in $\mathcal{M}$ of cardinality $d$ are exactly the ideals in $\mathcal{F}(d)$ as defined above.
Applying this kind of transformation for each player~$i$, we can identify the strategy set $\B_{f^{(i)}}(d_i)$ of each player~$i$ with the set family $\F_i(d_i)$, and this set family, in turn, with the matroid $\M_i = (E_i,r_i)$. 

With this notation, let us now return to the problem of finding a best response $\vec x_i$ of player~$i$ towards a strategy profile $\vec a=\vec x_{-i}\in \N^R$ of the remaining players. Note that, for $\vec a\in \N^R$, the player-specific strongly semi-convex cost functions $c_{i,r}: \N \to \N$
induce weight functions $w_i^{\vec a}:E\to \N$ on the ground set $E$ constructed above via
\begin{align*}
w_i^{\vec a}(r_t):= t c_{i,r}(a_r+t) - (t-1) c_{i,r} (a_r+t-1) \quad \text{ for all } r\in R, t \in \{1, \dots, d_i\}.
\end{align*}
Hence, finding a best response $\vec x_i \in \B_{f^{(i)}}(d_i)$ reduces to the problem of minimizing a linear function over the independent sets of cardinality $d_i$ of the matroid $\mathcal{M}_i=(E_i, r_i)$ associated with the submodular function $f^{(i)}$.

However, the ground set $E_i$ can be of exponential size, so that it is not a priori clear whether the matroid greedy algorithm
minimizes a linear weight function over the base polyhedron $\mathcal{B}_{f^{(i)}}(d_i)$ in strongly polynomial time.
Still, since we assume that the cost functions $c_{i,r} : \N \to \N$ are strongly semi-convex, it follows that the induced weight functions
$w_i^{\vec a}:E\to \N$ are \emph{admissible} in the sense that $e\prec g$ implies $w_i(e)\le w_i(g)$.

Given an ideal $F\in \mathcal{D}(P)$, we denote by
$F^+$  the set of $\preceq$-maximal elements in $F$, and by
 $(E \setminus F)^-$ the set of $\preceq$-minimal elements in $E \setminus F$. 
For $\mathcal{F}(d)$  as defined in \eqref{B}, and any admissible weight functions $w : E\to \N$,  
Faigle \cite{faigle1979greedy} showed that the following \emph{ordered greedy algorithm}
determines an ideal of minimal weight in $\mathcal{F}(d)$ (provided $\mathcal{F}(d)\neq \emptyset)$:

\begin{algorithm}[\textwidth]
 \caption{Ordered Greedy Algorithm}
 \label{alg:ordered_greedy}
 \Indm\Indmm
  \Indp\Indpp
$F \leftarrow \emptyset$\;
		\For{$k=1,\dots,d$}{	
			Let $e_k \leftarrow \arg\min \{w(e) : e \in (E \setminus F)^- \text{ and } F + e \in \mathcal{F}\}$\;
			$F \leftarrow F + e_k$\;
		}
Return $\mathcal{F}$\;
\end{algorithm}

In fact, the greedy algorithm determines in each iteration $k\le d$ an ideal of minimal weight in $\mathcal{F}(k)$.
The following proposition arises as a consequence of the discussion above and implies Lemma~\ref{lem:dem}.

\begin{proposition} 
\label{l.demand-increase}
Let $\F\subseteq \mathcal{D}(P)$ as defined in (\ref{F}), $k\in \N$, and $w:E\to\R$ admissible.
Suppose $F$ is of minimal $w$-weight in $\mathcal{F}(k)$. Then there exists $e\in E\setminus{F}$ such that
$F+e$ is of minimal $w$-weight in $\mathcal{F}(k+1)$
\end{proposition}

Due to the reduction of integral polymatroids to ordinary matroids, most of the structural properties of matroids carry over to integral polymatroids.
For example, the following proposition follows as a consequence of the well-known fact,
that for any basis $B$ of an ordinary matroid $\mathcal{M}$ which is not of minimal weight, there exists a local improvement step
towards a basis $B-e+f$ of smaller weight.

\begin{proposition}\label{l.local-improvement}
Suppose $F\in \mathcal{F}(k)$ is not of minimal $w$-weight for some admissible function $w:E\to\R$.
Then there exists some local improvement step
$F\rightarrow F-e+g\in \mathcal{F}(k)$ such that $w(e)>w(g)$.
\end{proposition}

The possibility to improve a non-optimal basis by local steps, 
as well as the possibility to uncross tight constraints due to the submodularity of the rank functions,
are the main ingredients of the proof of the following theorem, which implies Lemma~\ref{lem:load}.

\begin{theorem}\label{t.load-increase}
Let $F\in \mathcal{F}(k)$ be of minimal weight w.r.t.\ the admissible weight function $w$.
If the weight function $\bar{w}$ differs from $w$ only on chain $K_{r^*}$ such that
\begin{align*}
\bar{w}(r_k)=
\begin{cases}
w(r_k), & \text{ if }r \neq r^*,\\
w(r_{k+1}), & \text{ else},
\end{cases}
\end{align*}
then there exists $F'=F-e+g\in \mathcal{F}(k)$  of minimal weight w.r.t.\ $\bar{w}$.
\end{theorem}

\begin{proof}
Let $F$ be of minimal $w$-weight in $\mathcal{F}(k)$ and $\bar{w}$ as described above.
If $F$ is also of minimal $\bar w$-weight,  the theorem follows using $e=g$.
If $F$ is not $\bar w$-minimal, then
by Proposition \ref{l.local-improvement}, there exists $e\in F^+$ and $g\in (E\setminus{F})^-$ with $\bar{w}(e)>\bar{w}(g)$ such that
$F-e+g\in \mathcal{F}(k)$.We choose the pair $\{e,g\}$ maximizing the improvement $\bar{w}(e)-\bar{w}(g)$.
Since $F$ is $w$-optimal, it follows that $e\in K_{r^*}$ and $g\notin K_{r^*}$.
We want to show that $F'=F-e+g$ is of minimal $\bar{w}$-weight.
Suppose not.
By Proposition \ref{l.local-improvement}, there exists $F''=F'-p+h\in \mathcal{F}(k)$ with $\bar{w}(p)>\bar{w}(h)$. 
It follows that $h\notin K_{r^*}$.
Lemmata \ref{l.1} and \ref{l.2} which are stated and proven in Appendix~\ref{appendix:load-increase} imply
$g\not\preceq h$, and either $F-p+h\in \mathcal{F}(k)$ or $F-e+h\in \mathcal{F}(k)$.

We proceed to show the desired contradiction. 
First, we claim that $F-p+h \notin \mathcal{F}(k)$. To see this, note that
if $F-p+h\in \mathcal{F}(k)$, then also $w(p)\le w(h)$ by the $w$-optimality of $F$.
On the other hand,
 $\bar{w}(p)>\bar{w}(h)=w(h)$, which implies $\bar{w}(p)>w(p)$.
 This is only possible if $p\in K_{r^*}$.
 However, since $p\in F'$ and $e\not \in F'$,
 it follows that $p\prec e$.
 Thus $F-p+h$ is not an ideal in poset $P$, and we reach a contradiction.
 
Hence, we have $F-p+h\notin \mathcal{F}(k)$, but $F-e+h\in \mathcal{F}(k)$.
  Note that $\bar{w}(h)\ge \bar{w}(g)$, as otherwise, $\bar{w}(F-e+h)<\bar{w}(F-e+g)$
  in contradiction to the choice of the pair $\{e,g\}$.
  Moreover, since $g,h\notin K_{r^*}$, it follows that
  $w(h)=\bar{w}(h)\ge \bar{w}(g)=w(g).$

 We proceed to show that $F-p+g\notin \mathcal{F}(k)$.
  Suppose not.
  Then, by $w$-optimality of $F$, $w(p)\le w(g)$.
  On the other hand, 
  $\bar{w}(p)>\bar{w}(h)=w(h)$,
  since $\bar{w}(F'-p+h)<\bar{w}(F')$.
  If $p\not\preceq e$, then
  $w(h)<\bar{w}(p)=w(p)\le w(g)$, a contradiction.
  Thus, $p\prec e$ implying that $\bar{w}(p)=w(e)>\bar{w}(h)=w(h)$,
  in contradiction to $F-e+h\in \mathcal{F}(k)$.
  
   Summarizing, we have
  $F-p+g\notin \mathcal{F}(k)$, $F-p+h\notin \mathcal{F}(k)$ and $g\not\preceq h$.
  We show that this is not possible.
  In case $p\prec e$, we have
  $w(e)=\bar{w}(p)>\bar{w}(h)=w(h)$.
  This is in contradiction to the $w$-optimality of $F$,
  as we assumed that $F-e+h\in \mathcal{F}(k)$.
  Hence $p\not\preceq e$, implying that
  $F-p+h$ and $F-p+h$ are ideals in $P$.
  
 Thus, as $F-p+g\notin \mathcal{F}(k)$ and $F-p+h\notin \mathcal{F}(k)$,
 there must be sets $S, T\subseteq R$ with
 $|\bigcup_{r\in S} K_r \cap (F-p+g)| > f(S)$
 and $|\bigcup_{r\in T} K_r \cap (F-p+h)| > f(T)$.
 Since $F$ is feasible, it follows that
 \begin{align*}
 f(S)< \Big|\bigcup_{r\in S} K_r \cap (F-p+g)\Big| &\le \Big|\bigcup_{r\in S} K_r \cap F\Big|+1\le  f(S)+1,\\
 f(T)< \Big|\bigcup_{r\in T} K_r \cap (F-p+h)\Big| &\le \Big|\bigcup_{r\in T} K_r \cap F\Big|\le f(T)+1.
 \end{align*}
 Thus, $g\in \bigcup_{r\in S} K_r$, $p\notin \bigcup_{r\in S} K_r$,
 $h\in \bigcup_{r\in T} K_r$, and $p\notin \bigcup_{r\in T} K_r$.
 Moreover, by the integrality of $f$, $|\bigcup_{r\in S} K_r \cap F|=f(S)$ and
 $|\bigcup_{r\in T} K_r \cap F|= f(T)$.
 By the submodularity of $f$, it is not hard to see that also for $S\cap T$ and $S\cup T$ we have
\begin{align}
\label{eq:star3}
\Bigl|\bigcup_{r\in S\cap T} K_r \cap F\Bigr| = f(S\cap T) \quad \mbox{and} \quad \Bigl|\bigcup_{r\in S\cup T} K_r \cap 
F\Bigr|=f(S\cup T).
\end{align}
 However, $F-e+g-p+h \in \mathcal{F}(k)$ implies  
 $\{e,g,p,h\}\cap \bigcup_{r\in S\cup T} K_r = \{g,e,h\}.$
 Thus, using \eqref{eq:star3}, we obtain
 $|\bigcup_{r\in S\cup T} K_r\cap (F-e+g-p+h)| = |\bigcup_{r\in S\cup T} K_r\cap F|+1=f(S\cup T)+1$,
 a  contradiction to $F-e+g-p+h \in \mathcal{F}(k)$.
 \qed\end{proof}

\bibliographystyle{splncs03}
\bibliography{../master-bib}

\begin{thebibliography}{10}
\providecommand{\url}[1]{\texttt{#1}}
\providecommand{\urlprefix}{URL }

\bibitem{Ackermann09}
Ackermann, H., R\"{o}glin, H., V\"{o}cking, B.: Pure {N}ash equilibria in
  player-specific and weighted congestion games. Theoret. Comput. Sci.
  410(17),  1552--1563 (2009)

\bibitem{Anshelevich08}
Anshelevich, E., Dasgupta, A., Kleinberg, J., Tardos, {\'E}., Wexler, T.,
  Roughgarden, T.: The price of stability for network design with fair cost
  allocation. SIAM J. Comput.  38(4),  1602--1623 (2008)

\bibitem{AntonakopoulosCSZ11}
Antonakopoulos, S., Chekuri, C., Shepherd, F.B., Zhang, L.: Buy-at-bulk network
  design with protection. Math. Oper. Res.  36(1),  71--87 (2011)

\bibitem{Beckmann56}
Beckmann, M., McGuire, C., Winsten, C.: Studies in the Economics and
  Transportation. Yale University Press, New Haven, CT, USA (1956)

\bibitem{Chen09}
Chen, H., Roughgarden, T.: Network design with weighted players. Theory Comput.
  Syst.  45(2),  302--324 (2009)

\bibitem{ChenRV10}
Chen, H.L., Roughgarden, T., Valiant, G.: Designing network protocols for good
  equilibria. SIAM J. Comput.  39(5),  1799--1832 (2010)

\bibitem{Dunkel08}
Dunkel, J., Schulz, A.: On the complexity of pure-strategy {N}ash equilibria in
  congestion and local-effect games. Math. Oper. Res.  33(4),  851--868 (2008)

\bibitem{faigle1979greedy}
Faigle, U.: The greedy algorithm for partially ordered sets. Discrete Math.
  28(2),  153--159 (1979)

\bibitem{FalkeHarks13}
von Falkenhausen, P., Harks, T.: Optimal cost sharing for resource selection
  games. Math. Oper. Res.  38(1),  184--208 (2013)

\bibitem{Fotakis05}
Fotakis, D., Kontogiannis, S., Spirakis, P.: Selfish unsplittable flows.
  Theoret. Comput. Sci.  348(2-3),  226--239 (2005)

\bibitem{Gairing11}
Gairing, M., Monien, B., Tiemann, K.: Routing (un-)splittable flow in games
  with player-specific linear latency functions. ACM Trans. Algorithms  7(3),
  1--31 (2011)

\bibitem{Harks:TARK}
Harks, T., Klimm, M.: Congestion games with variable demands. In: Apt, K. (ed.)
  Proc. 13th Conf. Theoret. Aspects of Rationality and Knowledge. pp. 111--120
  (2011)

\bibitem{Harks:existence}
Harks, T., Klimm, M.: On the existence of pure {Nash} equilibria in weighted
  congestion games. Math. Oper. Res.  37(3),  419--436 (2012)

\bibitem{Haurie85}
Haurie, A., Marcotte, P.: On the relationship between {Nash-Cournot and
  Wardrop} equilibria. Networks  15,  295--308 (1985)

\bibitem{helgason1974aspects}
Helgason, T.: Aspects of the theory of hypermatroids. In: Hypergraph Seminar.
  pp. 191--213. Springer (1974)

\bibitem{Ieong05}
Ieong, S., McGrew, R., Nudelman, E., Shoham, Y., Sun, Q.: Fast and compact: A
  simple class of congestion games. In: Proc. 20th Natl. Conf. Artificial
  Intelligence and the 17th Innovative Appl. Artificial Intelligence Conf. pp.
  489--494 (2005)

\bibitem{Johari06}
Johari, R., Tsitsiklis, J.N.: A scalable network resource allocation mechanism
  with bounded efficiency loss. IEEE J. Sel. Area Commun.  24(5),  992--999
  (2006)

\bibitem{Kelly98}
Kelly, F., Maulloo, A., Tan, D.: Rate control in communication networks: Shadow
  prices, proportional fairness, and stability. J. Oper. Res. Soc.  49,
  237--252 (1998)

\bibitem{KrystaSV03}
Krysta, P., Sanders, P., V{\"o}cking, B.: Scheduling and traffic allocation for
  tasks with bounded splittability. In: Rovan, B., Vojtas, P. (eds.) Proc. 28th
  Internat. Sympos. Math. Foundations of Comput. Sci. LNCS, vol. 2747, pp.
  500--510 (2003)

\bibitem{Meyers08}
Meyers, C.: Network Flow Problems and Congestion Games: Complexity and
  Approximation Results. Ph.D. thesis, MIT, Operations Research Center (2006)

\bibitem{Milchtaich96}
Milchtaich, I.: Congestion games with player-specific payoff functions. Games
  Econom. Behav.  13(1),  111--124 (1996)

\bibitem{Milchtaich06}
Milchtaich, I.: The equilibrium existence problem in finite network congestion
  games. In: Mavronicolas, M., Kontogiannis, S. (eds.) Proc. 2nd Internat.
  Workshop on Internet and Network Econom. LNCS, vol. 4286, pp. 87--98 (2006)

\bibitem{Rosenthal73a}
Rosenthal, R.: A class of games possessing pure-strategy {N}ash equilibria.
  Internat. J. Game Theory  2(1),  65--67 (1973)

\bibitem{Rosenthal73b}
Rosenthal, R.: The network equilibrium problem in integers. Networks  3,
  53--59 (1973)

\bibitem{Roughgarden_Book2005}
Roughgarden, T.: Selfish Routing and the Price of Anarchy. MIT Press,
  Cambridge, MA, USA (2005)

\bibitem{schrijver2003combinatorial}
Schrijver, A.: Combinatorial optimization: polyhedra and efficiency, vol.~24.
  Springer (2003)

\bibitem{Srikant03}
Srikant, R.: The Mathematics of {Internet} Congestion Control. Birkh{\"a}user,
  Basel, Switzerland (2003)

\bibitem{Tran11}
Tran-Thanh, L., Polukarov, M., Chapman, A., Rogers, A., Jennings, N.: On the
  existence of pure strategy {Nash} equilibria in integer-splittable weighted
  congestion games. In: Persiano, G. (ed.) Proc. 4th Internat. Sympos.
  Algorithmic Game Theory. LNCS, vol. 6982, pp. 236--253 (2011)

\bibitem{Wardrop52}
Wardrop, J.: Some theoretical aspects of road traffic research. Proc. Inst.
  Civil Engineers  1(Part II),  325--378 (1952)

\end{thebibliography}

\newpage
\appendix
\renewcommand\thelemma{\thesection\arabic{lemma}}
\setcounter{lemma}{0}
\section{Appendix}
\subsection{Proof of Proposition~\ref{pro:semi_convex}}
\begin{proof}
We start to show the first part of the claim. Let $c : \N \to \N$ be a nondecreasing convex function. As marginal 
differences are nondecreasing, we obtain
\begin{align}
\label{eq:1}
c(a+x) - c(a+x-1) \leq c(a+y) - c(a+y-1)
\end{align}
for all $a \in \N$ $x,y \in \N$ with $1 \leq x \leq y$. Increasing $a$ does not decrease the right hand side of \eqref{eq:1} 
and, thus,
\begin{align}
\label{eq:2}
c(a+x) - c(a+x-1) \leq c(b+y) - c(b+y-1)
\end{align}
for all $a,b \in \N$ with $a \leq b$ and $x,y \in \N$ with $1 \leq x \leq y$. As both differences are nonnegative and $x 
\leq y$, we obtain a valid inequality when multiplying the left hand side of \eqref{eq:2} with $x$ and the right hand side 
of \eqref{eq:2} with $y$. This implies 
\begin{align}
\label{eq:3}
c(a+x)x - c(a+x-1)x \leq c(b+y)y - c(b+y-1)y.
\end{align}
Finally, we have $c(a+x-1) \leq c(b+y-1)$ as $c$ is nondecreasing. Adding this inequality to \eqref{eq:3} gives the claimed 
result.

To see that there is a strongly semi-convex function that is non convex, consider the function $c : \N \to \N$ defined as 
$c(x) = x$, if $x \leq 2$ and $c(x) = x - 1/4$, if $x \geq 3$. We proceed to show that $c$ is strongly semi-convex. For the 
following calculations, for an event $E$, we write $\chi_E$ for the indicator variable of event $E$, i.e., $\chi_E = 1$, if 
$E$ is true, and $\chi_E = 0$, otherwise. We calculate
\begin{align*}
&x \bigl(c(a+x)x - c(a+x-1)\bigr) + c(a+x-1)\\
&= x\Bigl(1- \frac{\chi_{a+x=3}}{4}\Bigr) + a + x - 1 - \frac{\chi_{a+x \geq 4}}{4}.
\end{align*}
Thus, $c$ is strongly semi-convex if and only if
\begin{align*}
x\Bigl(1- \frac{\chi_{a+x=3}}{4}\Bigr) + a + x - \frac{\chi_{a+x \geq 4}}{4} \leq y\Bigl(1- \frac{\chi_{b+y=3}}{4}\Bigr) + b 
+ y - \frac{\chi_{b+y \geq 4}}{4}
\end{align*}
for all $a,b,x,y \in \N$ with $a \leq b$ and $1 \leq x \leq y$. Clearly, if $a=b$ and $x=y$, this inequality is satisfied 
trivially, so we may assume that either $y \geq x+1$ or $b \geq a+1$. To prove strong semi-convexity, it is sufficient to 
show
\begin{align*}
2x + a + x \leq 2y + b - y \cdot \frac{\chi_{b+y=3}}{4} - \frac{1}{4}.
\end{align*}
As $b \geq 0$, we obtain $y \cdot \chi_{b+y=3}/4 \leq 3/4$. Using that either $y \geq x+1$ or $b \geq a+1$, we obtain the 
desired result.
\qed\end{proof}

\subsection{Proof of Corollary~\ref{cor:runtime}}\label{apx:cor}
\begin{proof}
We analyze the worst-case runtime of Algorithm~\ref{alg:greedy}. To this end, let us fix an iteration of the for-loop. In the proof of Theorem~\ref{thm:main}, we showed that during this iteration, for each player, the sorted vector of marginal costs (as defined in \eqref{def:delta}) decreases lexicographically during the while-loop. Moreover, the marginal cost of a particular unit of demand $i_j$ of player~$i$ assigned to a resource $r$ does not depend on the aggregated demand $\sum_{j \in N} x_{j,r}$ of all players for resource $r$, but only on the number of units of demand $x_{i,r}$ assigned to $r$ by player~$i$. We derive that for each player~$i$ and each resource $r$ at most $d_i$ different marginal cost values can occur. This observation bounds the number of different marginal cost vectors of player~$i$ by $(m \cdot d_i)^{d_i}$, where $m = |R|$. Since the marginal cost vectors lexicographically decrease, the total number of iterations of the while-loop for each iteration of the for-loop is bounded by $\sum_{i \in N} (m \cdot d_i)^{d_i}$. Setting $\dmax = \max_{i \in N} d_i$, this expression is bounded by $(n\cdot m \cdot \dmax)^{\dmax}$, where $n = |N|$. Using that there are $\sum_{i \in N} d_i \leq n \cdot \dmax$ iterations of the for-loop, one for each unit of demand in the game, we obtain the following corollary.
\end{proof}

\subsection{Proofs of Lemmata \ref{l.1} and \ref{l.2} needed to prove Theorem~\ref{t.load-increase}}
\label{appendix:load-increase}

In the proof of Theorem \ref{t.load-increase}, it is claimed that
$g\not\preceq h, \text{ and either } F-p+h\in \mathcal{F}(k) \text{ or } F-e+h\in \mathcal{F}(k).$
This  follows as a consequence of the following two Lemmata.

\begin{lemma}\label{l.1}
Let $F\in \mathcal{F}(k)$ be of minimal weight w.r.t.\ the admissible weight function $w$,
and consider function $\bar{w}$ as described in Theorem~\ref{t.load-increase}.
Suppose there exists $F'=F-e+g\in \mathcal{F}(k)$ with $\bar{w}(F')<\bar{w}(F)$,
and $F''=F'-p+h\in \mathcal{F}(k)$ with $\bar{w}(F'')<\bar{w}(F')$.
Then $g\not\preceq h$.
\end{lemma}

\begin{proof}
If $g\preceq h$, then
$\bar{w}(p)>\bar{w}(h)\ge w(h) \ge w(g)$.
We consider first the case where $p\preceq e \in K_{r^*}$.
Then, by construction of $\bar{w}$ we have
$w(e)=\bar{w}(p)>\bar{w}(h)=w(h)\ge w(g)$,
implying $w(F-e+g)<w(F)$ in contradiction to the $w$-optimality of $F$.

Thus, $p$ and $e$ are incomparable.
Since $w(p)> w(g)$,
the $w$-optimality of $F$ implies $F-p+g\notin \F(k)$.
But $F-p+g$ is an ideal in $P$ (as $p\not\preceq e$), so there must exist some set $S\subseteq R$ with
$|\bigcup_{r\in S} K_r \cap (F-p+g)|=f(S)+1$ and $|\bigcup_{r\in S} K_r \cap F|=f(S)$.
It follows that $g\in \bigcup_{r\in S} K_r$, and  $p\notin \bigcup_{r\in S} K_r$.

On the other hand,
we know that $F-e+g-p+h\in \F(k)$, implying that $e\in \bigcup_{r\in S} K_r$, and $h\notin \bigcup_{r\in S} K_r$.
However, $g\in \bigcup_{r\in S} K_r$ and $h\notin \bigcup_{r\in S} K_r$ is a contradiction to $g\preceq h$.
\qed\end{proof}

\begin{lemma}\label{l.2}
Suppose that $F\in \F(k), F'=F-e+g\in \F(k)$, $F''=F'-p+h\in \F(k)$, and $g\not \preceq h$.
Then either $F-p+h\in \F(k)$ or $F-e+h\in \F(k)$.
\end{lemma}

\begin{proof}
Consider first the case where $p\prec e$.
We know that $F-e+h$ is an ideal in $P$, since $g\not\preceq h$.
For the sake of contradiction, we assume that
$F-e+h\notin \F(k)$.
Then there exists some set $S\subseteq R$ with
$|\bigcup_{r\in S} K_r \cap (F-e+h)|=f(S)+1$.
It follows that $h\in \bigcup_{r\in S} K_r$, and  $e\notin \bigcup_{r\in S} K_r$.
Since we assume $e$ and $p$ to lie in the same chain, it follows that $p\notin \bigcup_{r\in S} K_r$.
However, this implies
$|\bigcup_{r\in S} K_r \cap (F-e +g -p +h)|\ge f(S)+1$, in contradiction to $F''=F-e+g-p+h\in \F(k)$.

It follows that $p\not\preceq e$, so that both, $F-p+h$ and $F-e+h$ are ideals in $P$.
Now we can use the same arguments as above: if neither
$F-p+h$, nor $F-e+h$ belongs to $\F(k)$, there exists a set $S\subseteq R$ with
$|\bigcup_{r\in S} K_r \cap (F-p+h)|=f(S)+1$ and $|\bigcup_{r\in S} K_r \cap F|=f(S)$,
implying $h\in \bigcup_{r\in S} K_r$, and  $p\notin \bigcup_{r\in S} K_r$,
and there also exists a set $T\subseteq R$ with
$|\bigcup_{r\in T} K_r \cap (F-e+h)|=f(T)+1$ and $|\bigcup_{r\in S} K_r \cap F|=f(T)$,
implying $h\in \bigcup_{r\in T} K_r$, and  $e\notin \bigcup_{r\in T} K_r$.
Since $|\bigcup_{r\in S} K_r \cap F|=f(S)$ and $|\bigcup_{r\in T} K_r \cap F|=f(T)$, ich follows by the submodularity of $f$ that
$|\bigcup_{r\in S\cup T} K_r \cap F|=f(S\cup T).$
Since
$\bigcup_{r\in S\cup T} K_r \cap \{e,g,p,h\}=\{p,h,e\}$, we have
$|\bigcup_{r\in S\cup T} K_r \cap (F-e+g-p+h)|\ge f(S\cup T)+1.$ A  contradiction to
$F''=F-e+g-p+h\in \F(k)$.
\qed\end{proof}

\end{document}